\documentclass[11pt]{article}

\usepackage{amssymb}
\usepackage{amsmath}
\usepackage{amsthm}
\usepackage{latexsym}
\usepackage{graphicx}
\usepackage{color}
\usepackage{graphics}
\usepackage[round]{natbib}
\usepackage{cancel}
\usepackage{mathtools}
\usepackage[mathscr]{euscript}
\usepackage{hyperref}
\usepackage{tikz}
\usetikzlibrary{decorations.markings}

\usepackage{fullpage}

\usepackage{MnSymbol}
\DeclareMathAlphabet\mathbb{U}{msb}{m}{n}
\usepackage{xpatch}

\DeclareMathOperator*{\E}{\mathbb E}

\DeclarePairedDelimiter{\abs}{\lvert}{\rvert}

\newtheorem{theorem}{Theorem}
\newtheorem{definition}{Definition}
\newtheorem{lemma}{Lemma}

\newcommand{\cF}{\mathcal{F}}

\newcommand{\cL}{\mathcal{L}}

\newcommand{\cP}{\mathcal{P}}

\newcommand{\cT}{\mathcal{T}}
\newcommand{\cV}{\mathcal{V}}

\newcommand{\bbR}{{\mathbb R}}

\newcommand{\ignore}[1]{}

\newcommand{\Circ}{\mathsf{Circ}}
\newcommand{\CircBar}{\overline{\mathsf{Circ}}}
\newcommand{\Arb}{\mathsf{Arb}}
\newcommand{\SArb}{\mathsf{SArb}}
\newcommand{\Flip}{\mathsf{Flip}}

\usepackage[algo2e,ruled]{algorithm2e}
\usepackage{algorithm}
\usepackage[noend]{algpseudocode}

\makeatletter
\newtheorem*{rep@theorem}{\rep@title}
\newcommand{\newreptheorem}[2]{%
\newenvironment{rep#1}[1]{%
 \def\rep@title{#2 \ref{##1}}%
 \begin{rep@theorem}}%
 {\end{rep@theorem}}}
\makeatother

\hypersetup{
  colorlinks   = true,
  urlcolor     = blue,
  linkcolor    = blue,
  citecolor    = blue
}

\newreptheorem{theorem}{Theorem}
\newreptheorem{lemma}{Lemma}
\newreptheorem{corollary}{Corollary}
\newreptheorem{proposition}{Proposition}

\title{Bernoulli Factories for Flow-Based Polytopes}
\author{Rad Niazadeh \\ Chicago Booth School of Business \and Renato Paes Leme \\ Google Research \and Jon Schneider \\ Google Research}
\date{}

\begin{document}

\maketitle

\begin{abstract}
We construct explicit combinatorial Bernoulli factories for the class of \emph{flow-based polytopes}; integral 0/1-polytopes defined by a set of network flow constraints. This generalizes the results of Niazadeh et al. (who constructed an explicit factory for the specific case of bipartite perfect matchings) and provides novel exact sampling procedures for sampling paths, circulations, and $k$-flows. In the process, we uncover new connections to algebraic combinatorics.
\end{abstract}

\section{Introduction}

A Bernoulli factory is a technique from applied probability that allows one to exactly sample certain random variables having only sample access to its parameters. Its original goal (\cite{asmussen1992stationarity,keane1994bernoulli}) was to perform exact simulation of certain stochastic processes, but it has since found applications in various fields such as mechanism design (\cite{dughmi2017bernoulli,cai2019efficient}), Bayesian inference (\cite{gonccalves2017exact,herbei2014estimating})  and quantum physics (\cite{dale2015provable,yuan2016experimental}).

Earlier work (\cite{von195113,keane1994bernoulli,nacu2005fast,mossel2005new}) on Bernoulli factories focused on sampling coins from other coins (e.g., given a coin with probability $p$, how can you construct a coin with probability $f(p)$?). Since then, there has been a lot of interest in exactly sampling random combinatorial objects such as: $k$-subsets (\cite{sampford1967sampling}), dice (\cite{morina2019bernoulli}), perfect matchings, and vertices of certain polytopes (\cite{niazadeh2021combinatorial}). In this paper, we provide factories for paths, flows and circulations uncovering new connections to algebraic combinatorics.

\paragraph{Sampling a path}
As an example of a combinatorial Bernoulli factory, consider the problem of sampling a path. We have a DAG with two special nodes which we call source and sink. On each edge $e$ of the DAG, sits a $p_e$-coin. A $p$-coin is a Bernoulli random variable with bias $p$, i.e., it outputs $1$ with probability $p$ and $0$ with probability $1-p$. Importantly, we \emph{do not know} the probabilities $p_e$ but we are promised that for every node other than the source and the sink, the sum of incoming probabilities equals the sum of outgoing probabilities. For the source, there are no incoming edges and the sum of outgoing probabilities is $1$. For the sink, there are no outgoing edges and the sum of incoming probabilities is $1$. We can flip the coins to obtain as many i.i.d. samples as we want. The task is to sample a path from the source to the sink such that each edge $e$ is included in the path with probability $p_e$ (exactly).\\

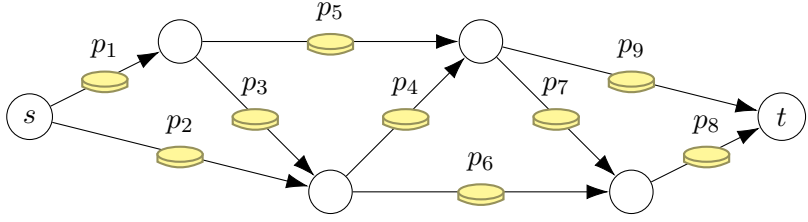
\begin{center}
\begin{tikzpicture}
\usetikzlibrary{arrows.meta}
\usetikzlibrary{calc}

\tikzset{vertex/.style = {shape=circle,draw,minimum size=1.5em}}
\tikzset{edge/.style = {-{Latex[length=3mm, width=2mm]}}}

% vertices
\node[vertex] (s) at  (0,0) {$s$};
\node[vertex] (a) at  (2,1) {};
\node[vertex] (b) at  (4,-1) {};
\node[vertex] (c) at  (6,1) {};
\node[vertex] (d) at (8,-1) {};
\node[vertex] (t) at (10,0) {$t$};
\draw[edge] (s) to (a);
\draw[edge] (s) to (b);
\draw[edge] (a) to (b);
\draw[edge] (b) to (c);
\draw[edge] (a) to (c);
\draw[edge] (c) to (t);
\draw[edge] (b) to (d);
\draw[edge] (c) to (d);
\draw[edge] (d) to (t);

\begin{scope}[shift={(1,.5)},scale=2]
\filldraw[color=yellow!50!black, fill=yellow!50!white, thick] (-0.15,0)--(-.15,-.05) to[out=-25,in=-155] (0.15,-0.05)--(.15,0)--cycle;
\filldraw[color=yellow!50!black, fill=yellow!50!white, thick] (0,0) ellipse (0.15 and 0.05);
\node at (0,.2) {$p_1$};
\end{scope}

\begin{scope}[shift={(2,-.5)},scale=2]
\filldraw[color=yellow!50!black, fill=yellow!50!white, thick] (-0.15,0)--(-.15,-.05) to[out=-25,in=-155] (0.15,-0.05)--(.15,0)--cycle;
\filldraw[color=yellow!50!black, fill=yellow!50!white, thick] (0,0) ellipse (0.15 and 0.05);
\node at (0,.2) {$p_2$};
\end{scope}

\begin{scope}[shift={(5,0)},scale=2]
\filldraw[color=yellow!50!black, fill=yellow!50!white, thick] (-0.15,0)--(-.15,-.05) to[out=-25,in=-155] (0.15,-0.05)--(.15,0)--cycle;
\filldraw[color=yellow!50!black, fill=yellow!50!white, thick] (0,0) ellipse (0.15 and 0.05);
\node at (0,.2) {$p_4$};
\end{scope}

\begin{scope}[shift={(3,0)},scale=2]
\filldraw[color=yellow!50!black, fill=yellow!50!white, thick] (-0.15,0)--(-.15,-.05) to[out=-25,in=-155] (0.15,-0.05)--(.15,0)--cycle;
\filldraw[color=yellow!50!black, fill=yellow!50!white, thick] (0,0) ellipse (0.15 and 0.05);
\node at (0,.2) {$p_3$};
\end{scope}

\begin{scope}[shift={(4,1)},scale=2]
\filldraw[color=yellow!50!black, fill=yellow!50!white, thick] (-0.15,0)--(-.15,-.05) to[out=-25,in=-155] (0.15,-0.05)--(.15,0)--cycle;
\filldraw[color=yellow!50!black, fill=yellow!50!white, thick] (0,0) ellipse (0.15 and 0.05);
\node at (0,.2) {$p_5$};
\end{scope}

\begin{scope}[shift={(6,-1)},scale=2]
\filldraw[color=yellow!50!black, fill=yellow!50!white, thick] (-0.15,0)--(-.15,-.05) to[out=-25,in=-155] (0.15,-0.05)--(.15,0)--cycle;
\filldraw[color=yellow!50!black, fill=yellow!50!white, thick] (0,0) ellipse (0.15 and 0.05);
\node at (0,.2) {$p_6$};
\end{scope}

\begin{scope}[shift={(7,0)},scale=2]
\filldraw[color=yellow!50!black, fill=yellow!50!white, thick] (-0.15,0)--(-.15,-.05) to[out=-25,in=-155] (0.15,-0.05)--(.15,0)--cycle;
\filldraw[color=yellow!50!black, fill=yellow!50!white, thick] (0,0) ellipse (0.15 and 0.05);
\node at (0,.2) {$p_7$};
\end{scope}

\begin{scope}[shift={(9,-0.5)},scale=2]
\filldraw[color=yellow!50!black, fill=yellow!50!white, thick] (-0.15,0)--(-.15,-.05) to[out=-25,in=-155] (0.15,-0.05)--(.15,0)--cycle;
\filldraw[color=yellow!50!black, fill=yellow!50!white, thick] (0,0) ellipse (0.15 and 0.05);
\node at (0,.2) {$p_8$};
\end{scope}

\begin{scope}[shift={(8,0.5)},scale=2]
\filldraw[color=yellow!50!black, fill=yellow!50!white, thick] (-0.15,0)--(-.15,-.05) to[out=-25,in=-155] (0.15,-0.05)--(.15,0)--cycle;
\filldraw[color=yellow!50!black, fill=yellow!50!white, thick] (0,0) ellipse (0.15 and 0.05);
\node at (0,.2) {$p_9$};
\end{scope}

\end{tikzpicture}
\end{center}

\noindent \emph{Solution:} Assume that the nodes are topologically sorted, i.e., the source is node $1$, the sink is node $n$ and every edge goes from $i$ to $j$ with $i < j$. For each node $u$, let $\textsf{In}(u)$ and $\textsf{Out}(u)$ be the incoming and outgoing edges of $u$. A path can be sampled as follows: start from node $u=1$, sample an outgoing edge $e = (u,v) \in \textsf{Out}(u)$ with probability proportional to $p_e$, add the edge to the path, set $u=v$ and repeat until $u$ is the sink. Sampling an outgoing edge proportionally to $p_e$ can be done only using sample access to the coins as follows: choose $e$ uniformly from  $\textsf{Out}(u)$ (using external randomness) and then flip the $p_e$-coin. If it returns $1$, choose this edge; otherwise, retry.

We now need to argue that each edge $e = (u,v)$ is chosen with probability $p_e$.  This probability is the probability that the process reaches node $u$ times the probability that the edge $e$ is sampled out of $u$. The latter probability is $p_e / \sum_{e' \in \textsf{Out}(u)} p_{e'}$. Let $\phi(u)$ be the probability that $u$ is reached. We are left to prove that $\phi(u) = \sum_{e' \in \textsf{Out}(u)} p_{e'}$. This follows by a simple induction. Observe that $\phi(s) = 1$. Now, assuming the induction hypothesis for every $v < u$:
$$\phi(u) = \sum_{(v,u) \in \textsf{In}(u)} \phi(v) \cdot \frac{p_{(v,u)}}{\sum_{(v,w) \in \textsf{Out}(v)} p_{(v,w)}} = \sum_{(v,u) \in \textsf{In}(u)} p_{(v,u)} = \sum_{e \in \textsf{Out}(u)} p_e $$\\

\paragraph{Sampling flows and circulations}
Sampling a path is relatively straightforward. Consider the more difficult version of the above problem where the sum of probabilities outgoing from the source and incoming to the sink is $2$, and we are asked to sample an integral $2$-flow (i.e., a pair of edge-disjoint paths from source to sink) such that each edge is in the flow with probability $p_e$. In another version, the sum of incoming probabilities equals the sum of outgoing probabilities for all nodes and we are asked to sample an integral circulation such that each edge $e$ is sampled with probability $p_e$. Solving these problems will require a new set of algebraic and combinatorial tools we develop in this paper.

% \paragraph{Bernoulli factories} The previous puzzle is an example of the type of problem known as a \emph{Bernoulli factory}: we are given a stream of i.i.d. samples from an unknown distribution and asked to sample an object from a related distribution. Earlier papers (\cite{von195113,keane1994bernoulli,nacu2005fast,mossel2005new}) focused on sampling coins from other coins. There has been a lot of interest in sampling random combinatorial objects such as: $k$-subsets (\cite{sampford1967sampling}), dice (\cite{morina2019bernoulli}), perfect matchings  and vertices in certain polytopes (\cite{niazadeh2021combinatorial}). In this paper, we provide factories for paths, flows, circulations and imperfect matchings uncovering new connections to algebraic combinatorics.\\

To understand the problem of sampling flows and circulations using a Bernoulli factory it is useful to consider the a general problem in which we are given polytope $\cP \subseteq [0,1]^n$ and $n$ coins with biases $p_1, \hdots, p_n$ with the promise that the vector $p = (p_1, \hdots, p_n)$ is in $\cP \cap (0,1)^n$. The biases themselves are \emph{unknown} but we are allowed to flip the coins as many times as we want. Our goal is to sample a vertex $v$ of the polytope such that $\E[v] = p$. 

\cite{niazadeh2021combinatorial} show that there exists a sampling procedure iff the polytope $\cP$ is the intersection of the hypercube $[0,1]^n$ and an affine subspace $\{x \in \bbR^n; Ax = b\}$ where $A$ is a full-rank $k \times n$ matrix and $b \in \bbR^k$ is a vector. The proof constructs a sampling algorithm for all \emph{generic} polytopes of this type. We say that a polytope
$$\cP = \{x \in [0,1]^n; Ax=b\}$$
given by a rank-$k$ matrix $A$ is generic iff all of its vertices have exactly $k$ coordinates in $(0,1)$. For any fixed matrix $A$, the set of vectors $b$ for which the polytope is not generic forms a measure zero set. Hence, every polytope can be made generic by slightly perturbing $b$.

Nevertheless, most relevant polytopes in combinatorial optimization (matching, flows, circulations, $k$-subset, etc.) are non-generic since the polytopes are integral and their vertices have all coordinates belonging $\{0,1\}$. Although \cite{niazadeh2021combinatorial} show that a sampling procedure also exists for \emph{non-generic} polytopes, their proof is \emph{non-constructive}. The proof proceeds as follows: for a non-generic polytope $\cP = \{x \in [0,1]^n; Ax=b\}$ they consider a sequence of vectors $b_t \rightarrow b$ such that the polytopes $\cP_t = \{x \in \bbR^n; Ax=b_t\}$ are generic. They show that each sampling procedure for $\cP_t$ can be encoded by an object in a large but finite space. Since the space is finite, one of those objects must appear \emph{infinitely often}. They can then use this object to design a factory for the limit polytope $\cP$. Another non-constructive proof is provided in \cite{leme2022multiparameter} based on concentration arguments.

Even though we know that a sampling procedure exists for all non-generic polytopes, we don't know what this algorithm looks like. For the case of sampling perfect bipartite matchings, an explicit sampling procedure is constructed in \cite{niazadeh2021combinatorial} and for sampling $k$-subsets, the classic paper by \cite{sampford1967sampling} gives a sampling procedure.

\paragraph{Our Results} In this paper we bridge this gap by constructing Bernoulli factories for all flow-like polytopes, which include circulations in a directed graph and $k$-flows in a DAG. We say that a polytope is flow-like if the variables can be indexed by edges in a directed graph with directed edge set $E$ and node set $N$ and the polytope is of the form:
$$\cP = \left\{ x \in [0,1]^E; \sum_{e \in \textsf{Out}(v)}  x_e - \sum_{e \in \textsf{In}(v)}  x_e = d_v, \forall v \in N \right\}$$
where $\textsf{Out}(v)$ and $\textsf{In}(v)$ are the outgoing and incoming edges to node $v$ and $d \in \bbR^N$ is a demand vector. Setting $d_v = 0$ for all $v \in N$ we obtain the circulation polytope whose vertices are $\{0,1\}$-circulations in the graph $(N,E)$. This polytope is integral since the constraint matrix is totally unimodular (see \cite{schrijver2003combinatorial} for example). If we choose a source $s$, a sink $t$ and set $d_s = k$, $d_t = -k$ and $d_v = 0$ for all $v \neq s,t$ we obtain the polytope whose vertices are the integral $k$-flows from $s$ to $t$.

The factory which performs this sampling is explicit and simple to describe, but requires a tiny amount of additional notation:

% Let $\cP$ be a flow-based polytope. We must construct a combinatorial Bernoulli factory for $\cP$. We first define some notation:

\begin{itemize}
    \item Given a directed edge $e$, let $\overline{e}$ denote the reverse of this edge; i.e., if $e = (u, v)$, then $\overline{e} = (v, u)$.
    \item A \textit{directed tree} is a collection of $n-1$ edges (in $E$) which, when viewed as undirected edges, form a tree. Some (but not all) directed trees are \textit{arborescences} -- directed trees where all edges point towards a specific root node. Let $\cT$ be the set of all directed trees (with edges in $E$), and let $\cT(E)$ be the set of all directed trees which use edges specifically in $E$.
    \item Let $f \in \{0, 1\}^{E}$ be a vertex of $\cP$. Then for any edge $e \in E$, we define $\Flip_{f}(e)$ to equal $e$ if $f_{e} = 0$ and $\overline{e}$ if $f_{e} = 1$. We (slightly abusing notation) extend this to directed trees by, for any $T \in \cT$, letting $\Flip_{f}(T)$ be the element of $\cT$ formed by replacing each edge $e$ in $T$ with $\Flip_{f}(e)$. (In other words, we toggle the direction of all edges in $T$ that are also set in $f$).  
\end{itemize} 

% \begin{figure}[h]
% \centering
% \begin{tikzpicture}[scale=.7,decoration={
%     markings,
%     mark=at position 0.5 with {\arrow{>}}}
%     ]

% \node[circle,fill,inner sep=1.5pt] at (0,0) {};
% \node[circle,fill,inner sep=1.5pt] at (2,0) {};
% \node[circle,fill,inner sep=1.5pt] at (2,2) {};
% \node[circle,fill,inner sep=1.5pt] at (0,2) {};
% \draw[postaction={decorate},line width=1pt] (0,2) -- (0,0);
% \draw[postaction={decorate},line width=1pt] (2,2) -- (0,2);
% \draw[postaction={decorate},line width=1pt] (2,0) -- (0,0);
% \node at (-.3,0) {$1$};
% \node at (2.3,0) {$2$};
% \node at (-.3,2) {$3$};
% \node at (2.3,2) {$4$};

% \begin{scope}[xshift=150]
% \node[circle,fill,inner sep=1.5pt] at (0,0) {};
% \node[circle,fill,inner sep=1.5pt] at (2,0) {};
% \node[circle,fill,inner sep=1.5pt] at (2,2) {};
% \node[circle,fill,inner sep=1.5pt] at (0,2) {};
% \draw[postaction={decorate},line width=1pt] (0,0) -- (2,0);
% \draw[postaction={decorate},line width=1pt] (2,2) -- (0,2);
% \draw[postaction={decorate},line width=1pt] (0,0) -- (0,2);
% \node at (-.3,0) {$1$};
% \node at (2.3,0) {$2$};
% \node at (-.3,2) {$3$};
% \node at (2.3,2) {$4$};
% \end{scope}

% \end{tikzpicture}
% \caption{First graph is an arborescence rooted at $1$. The second is a directed tree that is not an arborescence.}
% \label{fig:arb_3}
% \end{figure}

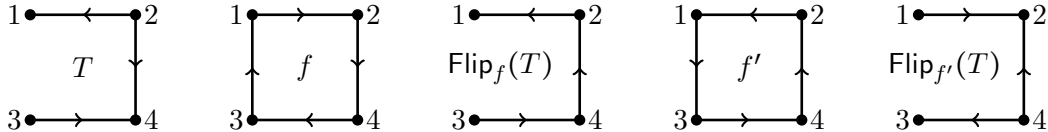
\begin{figure}[h]
\centering
\begin{tikzpicture}[scale=.7,decoration={
    markings,
    mark=at position 0.5 with {\arrow{>}}}
    ]

\node[circle,fill,inner sep=1.5pt] at (0,0) {};
\node[circle,fill,inner sep=1.5pt] at (2,0) {};
\node[circle,fill,inner sep=1.5pt] at (2,2) {};
\node[circle,fill,inner sep=1.5pt] at (0,2) {};
\draw[postaction={decorate},line width=1pt] (0,0) -- (2,0);
\draw[postaction={decorate},line width=1pt] (2,2) -- (0,2);
\draw[postaction={decorate},line width=1pt] (2,2) -- (2,0);
\node at (-.3,0) {$3$};
\node at (2.3,0) {$4$};
\node at (-.3,2) {$1$};
\node at (2.3,2) {$2$};
\node at (1,1) {$T$};

\begin{scope}[xshift=120]
\node[circle,fill,inner sep=1.5pt] at (0,0) {};
\node[circle,fill,inner sep=1.5pt] at (2,0) {};
\node[circle,fill,inner sep=1.5pt] at (2,2) {};
\node[circle,fill,inner sep=1.5pt] at (0,2) {};
\draw[postaction={decorate},line width=1pt] (0,0) -- (0,2);
\draw[postaction={decorate},line width=1pt] (0,2) -- (2,2);
\draw[postaction={decorate},line width=1pt] (2,2) -- (2,0);
\draw[postaction={decorate},line width=1pt] (2,0) -- (0,0);
\node at (-.3,0) {$3$};
\node at (2.3,0) {$4$};
\node at (-.3,2) {$1$};
\node at (2.3,2) {$2$};
\node at (1,1) {$f$};
\end{scope}

\begin{scope}[xshift=240]
\node[circle,fill,inner sep=1.5pt] at (0,0) {};
\node[circle,fill,inner sep=1.5pt] at (2,0) {};
\node[circle,fill,inner sep=1.5pt] at (2,2) {};
\node[circle,fill,inner sep=1.5pt] at (0,2) {};
\draw[postaction={decorate},line width=1pt] (0,0) -- (2,0);
\draw[postaction={decorate},line width=1pt] (2,2) -- (0,2);
\draw[postaction={decorate},line width=1pt] (2,0) -- (2,2);
\node at (-.3,0) {$3$};
\node at (2.3,0) {$4$};
\node at (-.3,2) {$1$};
\node at (2.3,2) {$2$};
\node at (.5,1) {$\Flip_f(T)$};
\end{scope}

\begin{scope}[xshift=360]
\node[circle,fill,inner sep=1.5pt] at (0,0) {};
\node[circle,fill,inner sep=1.5pt] at (2,0) {};
\node[circle,fill,inner sep=1.5pt] at (2,2) {};
\node[circle,fill,inner sep=1.5pt] at (0,2) {};
\draw[postaction={decorate},line width=1pt] (0,2) -- (0,0);
\draw[postaction={decorate},line width=1pt] (2,2) -- (0,2);
\draw[postaction={decorate},line width=1pt] (2,0) -- (2,2);
\draw[postaction={decorate},line width=1pt] (0,0) -- (2,0);
\node at (-.3,0) {$3$};
\node at (2.3,0) {$4$};
\node at (-.3,2) {$1$};
\node at (2.3,2) {$2$};
\node at (1,1) {$f'$};
\end{scope}

\begin{scope}[xshift=480]
\node[circle,fill,inner sep=1.5pt] at (0,0) {};
\node[circle,fill,inner sep=1.5pt] at (2,0) {};
\node[circle,fill,inner sep=1.5pt] at (2,2) {};
\node[circle,fill,inner sep=1.5pt] at (0,2) {};
\draw[postaction={decorate},line width=1pt] (2,0) -- (0,0);
\draw[postaction={decorate},line width=1pt] (0,2) -- (2,2);
\draw[postaction={decorate},line width=1pt] (2,0) -- (2,2);
\node at (-.3,0) {$3$};
\node at (2.3,0) {$4$};
\node at (-.3,2) {$1$};
\node at (2.3,2) {$2$};
\node at (.5,1) {$\Flip_{f'}(T)$};
\end{scope}

\end{tikzpicture}
\caption{The edges in $T$ form a direct tree but not an arborescence. When we flip the edges of $T$ using circulation $f$ we obtain an arborescence $\Flip_f(T)$ rooted at vertex $1$. On the other hand if we flip the edges in $T$ using circulation $f'$ we obtain a directed tree that is not an arborescence.}
\label{fig:flip}
\end{figure}

We illustrate these definitions in Figure \ref{fig:flip}. With this, we are ready to present the sampling algorithm: 

\begin{algorithm}[H]
\caption{{\sc Bernoulli Factory for Flow-like Polytopes}}
\label{alg:bernoulli_flow}
\begin{algorithmic} 
\State For each edge $e \in E$, flip the corresponding $p_e$-coin and set $f_e \in \{0,1\}$ equal to the outcome.
\State If $f$ does not form a valid flow (i.e., $f \notin \cP$), restart.
\State Choose a directed tree $T \in \cT(E)$ uniformly at random among all trees where $\Flip_{f}(T)$ is an arborescence rooted at vertex $1$. (We describe how to do this efficiently in Section \ref{sec:discussion}).
\State For each $e \in T$, flip the corresponding $p_e$-coin. If the outcome is equal to $f_e$, restart.
\State Output $f$.
\end{algorithmic}
\end{algorithm}

Our main result is that the above procedure successfully exactly samples a vertex of $\cP$ with the correct marginal probabilities (subject to some fairly weak conditions on the flow-like polytope $\cP$ and the interior point $p \in \cP$ which we discuss in Section \ref{sec:discussion}).

\begin{theorem}\label{thm:main}
Let $\cP$ be a flow-like polytope defined over a set $E$ of edges that are connected when viewed as undirected edges. If $p \in \cP \cap (0,1)^n$, then the Bernoulli factory in Algorithm \ref{alg:bernoulli_flow} samples a vertex $f$ of $\cP$ such that edge $e$ is included in the flow with probability $p_e$. 
\end{theorem}

To understand at a high level why Algorithm \ref{alg:bernoulli_flow} works (and why Theorem \ref{thm:main} is true), note that in the first two lines of Algorithm \ref{alg:bernoulli_flow} we sample a flow $f$ with probability proportional to $\prod_{e \in f} p_e \cdot \prod_{e \not\in f} (1-p_e)$. Superficially, this may seem like what we want (since edge $e$ is ``included'' with probability $p_e$ and ``excluded'' with probability $1-p_e$) but because the total mass of flows containing $e$ and not containing $e$ are not equal, this does not result in the correct marginals. The remaining lines of Algorithm \ref{alg:bernoulli_flow} can be thought of as applying a ``correction term'' to this sampling procedure that corrects the marginal distribution.

To actually prove Theorem \ref{thm:main} (and understand why the correction term is the way it is), it is useful to rephrase this sampling problem algebraically. Specifically, the above sampling procedure (and a large class of sampling procedures) can be thought of as sampling each vertex $f$ of $\cP$ with probability proportional to some polynomial $P_{f}(p)$. The constraints that the marginals are correct and that this sampling is implementable by a Bernoulli translate to algebraic constraints on the set of polynomials $P_{f}(p)$. The combinatorial and algebraic properties of arborescences allow us to construct a nice explicit solution to this set of algebraic constraints (of particular utility are various results stemming from the matrix-tree theorem, which allow us to root the arborescences at an arbitrary node; indeed, Algorithm \ref{alg:bernoulli_flow} works regardless of which node we label as ``vertex 1''). 

\paragraph{Why exact sampling?}
The Bernoulli factory problem originated from a question by \cite{asmussen1992stationarity} on how to \emph{exactly} simulate certain stochastic processes. In various situations obtaining samples is easy but understanding the \emph{exact} distribution they come from is difficult, for example when the samples are solutions to a stochastic differential equation or when they are the outcome of a simulation. Bernoulli factories provide a way to perform \emph{exact} computation having only samples, as opposed to precise-but-approximate computation. Approximate sampling is much simpler. One can flip each coin many times to obtain a good approximation to $p_e$ (by Chernoff bounds) and use those values to sample the desired object.

Exact sampling is essential in some of the applications. In mechanism design, it leads to reductions preserving exact Bayesian incentive compatibility, while approximate sampling leads to the weaker game theoretical guarantee of approximate Bayesian incentive compatibility~\citep{dughmi2017bernoulli}. In Bayesian inference, it prevents errors from compounding in an iterative procedure\citep{gonccalves2017exact,gonccalves2017barker}.

\section{Preliminaries}

\subsection{Multiparameter Bernoulli Factories}

We begin by formally defining multivariate Bernoulli factories.

\begin{definition}\label{def:factory}
A Bernoulli factory $\mathcal{F}$ with output in $V$ and input coins $(x_1,\hdots, x_n)$ corresponds to a (possibly infinite) rooted binary tree $\mathcal{T}$ where each node of $\mathcal{T}$ has either $2$ children (an internal node) or $0$ (a leaf). Each internal node is labelled either with a variable $x_i$ or with a constant $c \in (0,1)$. Each leaf is labelled with $v \in V$, representing the output of the factory upon reaching that node.

To execute the factory with coins $(x_1, \hdots, x_n)$ we start from the root and at each node we flip the coin given the label of that node (either one of the $x_i$-coins of unknown bias or a $c$-coin of known bias). Based on the outcome, we either take the left edge ($0$) or the right edge ($1$). If we reach a leaf, we output its label.
\end{definition}

A factory $\mathcal{F}$ is valid in a domain $D \subseteq [0,1]^n$ if for any input $x = (x_1, \hdots, x_n) \in D$ it reaches a leaf almost surely. Therefore we can view a factory as a random variable $\mathcal{F}$ taking values in $V$. The distribution of $\mathcal{F}$ will naturally depend on the input coins $x$.

\subsection{Bernstein Polynomials}

\cite{niazadeh2021combinatorial} define a way to construct Bernoulli factories using Bernstein polynomials, which we will also use here. A Bernstein monomial is an expression of the type $\prod_{i \in [n]} x_i^{a_i} (1-x_i)^{b_i}$ for non-negative integers $a_i$ and $b_i$. A Bernstein polynomial is a positive linear combination of Bernstein monomials $P(x) = \sum_i c_i M_i(x)$ where $c_i \geq 0$ and $M_i(x)$ are Bernstein monomials. The following lemma is shown in \cite{niazadeh2021combinatorial} and a proof is provided in Appendix \ref{appendix:bernstein} for completeness.

\begin{lemma}\label{lem:bernstein}
Given Bernstein polynomials $P_v(x)$ for each $v \in V$ such that $\sum_v P_v(x) \neq 0$, there is a Bernoulli factory that outputs each $v$ with probability $P_v(x) / \sum_{v'} P_{v'}(x)$. 
\end{lemma}

\begin{lemma}\label{lemma:bernstein}
Let $\cP \subseteq [0,1]^n$ be a polytope and $V$ its set of vertices. If for each $v \in V$ there is a Bernstein polynomial $P_v(x)$ such that $\sum_v (v-x) P_v(x) = 0$ and $\sum_v P_v(x) \neq 0$ then there exists a Bernoulli factory that for coins $x \in \cP \cap (0,1)^n$, it returns a vertex $v$ such that $\E[v] = x$.
\end{lemma}

\begin{proof}
By the previous lemma, there is a factory that samples each vertex $v$ with probability proportional to $P_v(x) / \sum_{v'} P_{v'}(x)$, hence:
$\E[v] = \sum_v v \cdot P_v(x) / (\sum_{v} P_{v}(x)) = x$.
\end{proof}

\subsection{Matrix-Tree Theorem}\label{sec:matrix_tree}

The Matrix-Tree Theorem relates the minors of the Laplacian matrix of a directed weighted graph with a polynomial defined over the arborescences of this graph. We begin by defining these concepts. Given a directed graph on $n$ nodes, an \textit{arborescence} rooted at $i$ is a collection $A$ of $(n-1)$ directed edges where from any vertex $v$ there is a unique walk from $v$ to $i$ using only the directed edges in $A$ (i.e., arborescences are trees directed towards $i$). Let $\Arb_{i}$ be the set of arborescences rooted at $i$.

Now we define the Laplacian. Given edge weights  $[x_{i,j}]_{n\times n}$, the $n$-by-$n$ weighted directed Laplacian matrix associated with this graph is defined as:
$$
\forall i,j\in[n]:~\mathcal{L}_{i,j} =
\left\{
	\begin{array}{ll}
		\sum_{k\neq i}x_{k,i}  & \mbox{if } i=j \\
		-x_{i,j} & \mbox{if } i\neq j
	\end{array}
\right.
$$
For any $r\in[n]$, let $\mathcal{L}^{(r)}$ denote the $(n-1)$-by-$(n-1)$ submatrix of $\mathcal{L}$ obtained by removing the row and the column corresponding to $r$. 

\begin{theorem}[\emph{Matrix-Tree Theorem~\citep{tutte1948dissection}}]\label{thm:tutte}
Let $\mathcal{L}$ be the $n$-by-$n$ weighted directed Laplacian matrix with weights $[x_{i,j}]$. Then for any $r\in[n]$, 
$$
\det[\mathcal{L}^{(r)}]=\sum_{A \in \Arb_{r}(n)}\prod_{(u, v) \in A} x_{u, v}.
$$
\end{theorem}

We will use the Matrix-Tree theorem together with a result about the principal minors of a \emph{zero-line-sum (ZLS)} matrix, i.e., a matrix whose rows and columns all sum to 0, are equal. We refer to \cite{niazadeh2021combinatorial} for a proof.

\begin{lemma}[\emph{ZLS matrices have equal cofactors}]\label{thm:zls}
Let $A$ be an $n$-by-$n$ matrix such that  $\sum_{j=1}^{n} A_{ij} = \sum_{j=1}^{n} A_{ji} = 0$ for each $i \in [n]$. If $A^{(r)}$ is the matrix obtained by removing the $r$-th column and $r$-th row then:
$$\det[A^{(r)}] = \det[A^{(r')}]$$
\end{lemma}

\subsection{Flow-Based Polytopes}

We will define a class of polytopes that we call \textit{flow-based polytopes}.

Consider a complete directed graph $G$ on $n$ vertices (so the set of all possible edges is the product set $[n] \times [n]$, i.e. all ordered pairs of vertices). We will choose a subset $E \subseteq E_{0} = \{(u, v) ; u \in [n], v \in [n], u \neq v\}$ of non-loop edges and call these edges \textit{variable edges} ($E_0$ is the base set of all possible non-loop edges). We also define a \textit{demands} $d_v$ for each $v \in [n]$ associating vertices to integer ``demands'' (technically, these are negative demands).

Given $E$ and $d$, we can define a flow-based polytope. Our polytope will have $|E|$ variables, which we will write as $x_e$ (indexed by elements of $E$). Each $x_{e}$ must lie in $[0, 1]$. In addition, we will impose the following $n$ constraints, for each $v \in [n]$:

\begin{equation}\label{eq:affine}\sum_{i}x_{vi} - \sum_{i}x_{iv} = d_v.\end{equation}

\noindent
Some edges in the above sums may not belong to $E$; in that case we abuse notation and assume $x_{e} = 0$ where these terms appear. One interpretation of the above constraints is that we are looking for a ``flow'' in $G$, where each directed edge $e$ has $x_{e}$ units of flow and the net flow out of each vertex $v$ must equal $d_v$.

Let $\cP \subseteq [0, 1]^{E}$ denote the resulting polytope. Since flow constraints are totally unimodular, $\cP$ has integral 0/1-vertices. We will let $\cV(\cP)$ denote the set of vertices of $\cP$. We will generally use the letter $f$ to denote an individual element of $\cV(\cP)$ (think ``flow''; unfortunately, $v$ is already reserved for vertices of $G$). 
Some natural examples of flow-based polytopes:

\begin{itemize}
    \item \textbf{Circulations}: Let $E = E_{0} = \{(u, v); u \in [n], v \in [n], u \neq v\}$ (all non-loops) and $d_v = 0$ for all $v \in [n]$. The resulting polytope is the \textbf{circulation polytope}. Vertices of this polytope correspond to \textbf{circulations}: directed graphs where each vertex has equal in-degree and out-degree.
    \item \textbf{Bipartite Perfect Matchings}: Let $n = 2m$. Let $E = \{(u, v+m); u \in [m], v \in [m]\}$ and
    
    $$d_v = \begin{cases}
    1 & \mbox{ if } 1 \leq u \leq m \\
    -1 & \mbox{ if } m+1 \leq u \leq 2m.
    \end{cases}$$
    
    The resulting polytope is the \textbf{perfect matching polytope} (alternately known as the Birkhoff-von Neumann polytope). Vertices of this polytope correspond to \textbf{perfect matchings} in the bipartite graph $G$.
    \item \textbf{Flows in a DAG}. Let $E = \{(u, v); u \in [n], v \in [n], u < v\}$ (edges from a smaller labeled vertex to a larger labeled vertex). Fix an integer $k > 0$, and let 
    
    $$d_v = \begin{cases}
    k & \mbox{ if } u = 1 \\
    -k & \mbox{ if } u = n\\
    0 & \mbox{ if } 1 < u < n.
    \end{cases}$$
    
    The resulting polytope is the \textbf{$k$-flow polytope}. Vertices of this polytope correspond to \textbf{unions of $k$ edge-disjoint paths} between $1$ and $n$. 
\end{itemize}

\subsection{The Circulation Polytope}

The circulation polytope mentioned above will end up being particularly important in the analysis that follows, so we will give it a special name -- $\Circ$ -- and discuss it a little further here.

The circulation polytope $\Circ$ is the $n(n-1)$-dimensional polytope in variables $x_{e}$ (where $e \in \{(u, v) ; u \in [n], v\in[n], u \neq v\}$) with constraints $x_{e} \in [0, 1]$ and (for all $v \in [n]$)

\begin{equation} \label{eq:circbar}
    \sum_{i \neq v} x_{vi} - \sum_{i \neq v} x_{iv} = 0.
\end{equation}

Sometimes we will want to consider just the hyperplane defined by \eqref{eq:circbar} above (without intersecting it with the hypercube $[0, 1]^{E_{0}}$). We will call this hyperplane $\CircBar$ (note that $\CircBar$ can also be thought of as the affine span of $\Circ$). Similarly, we will let $\overline{\cP}$ denote the set of points satisfying just equation \eqref{eq:affine} (i.e., the affine span of $\cP$). Note that all flow-based polytopes are ``parallel'' to the circulation polytope in the following sense.

\begin{lemma}
Let $\cP$ be a flow-based polytope, and let $x, x'$ be points in $\cP$. Then $x - x' \in \CircBar$.\footnote{Technically $x-x'$ and $\CircBar$ lie in different spaces, but there is a natural embedding from $[0, 1]^{E}$ into $[0, 1]^{E_{0}}$.}
\end{lemma}
\begin{proof}
Since $x$ and $x'$ lie in $\cP$, it is true that: 
$\sum_{i}x_{vi} - \sum_{i}x_{iv} = d_{v},$
and
$\sum_{i}x'_{vi}  - \sum_{i}x'_{iv} = d_{v}.$
Subtracting these two equations gives
$\sum_{i}(x_{vi} - x'_{vi}) - \sum_{i}(x_{iv} - x'_{iv}) = 0,$
so $x -x'$ satisfies \eqref{eq:circbar} and $x - x' \in \CircBar$.
\end{proof}

\section{Bernoulli Factory Construction}

% Our main result is that Algorithm \ref{alg:bernoulli_flow} samples a flow $f$ (more precisely, a vertex $f$ of a flow-like polytope $\cP$) such that edge $e$ is included in the flow with probability $p_e$.

We now analyze Algorithm \ref{alg:bernoulli_flow} and show that it samples a vertex $f$ of $\cP$ such that each edge $e$ is included in the flow with probability $p_e$ as stated in Theorem \ref{thm:main}. Our proof will be algebraic and based on the structure of certain Bernstein polynomials. Before defining them, we recall a few definitions from the introduction:

\begin{itemize}
    \item A directed tree is a collection of $n-1$ edges that form a tree when viewed as undirected edges. The set $\cT(E)$ is the set of all directed trees with edges in $E$. We omit $E$ when clear from the context.
    \item If for a given directed tree $T \in \cT$ there is a vertex $r$ such that there is a walk in $T$ from $i$ to $r$ for every node $i$, we say that $T$ is an arborescence rooted at $r$. The set $\Arb_r$ is the set of all arborescences rooted at $r$.
    \item Given a vertex $f$ of $\cP$ and an edge $e \in E$, we define $\Flip_f(e) = e$ if $f_e = 0$ and $\Flip_f(e) = \bar e$ (the reverse of $e$) if $f_e = 1$. For a directed tree $T$ we define $\Flip_f(T) = \{\Flip_f(e); e \in T\}$.
\end{itemize}

Now define the following polynomials for each vertex $f$ of $\cP$ and root vertex $r \in [n]$.
% We can now define the factory. Recall that we can specify a combinatorial Bernoulli factory by specifying a polynomial $P_{f}(x) : \cP \rightarrow \bbR^{\geq 0}$ for each vertex $f$ of $\cP$. Here is the polynomial for a flow-based polytope (fix an arbitrary root vertex $r$):

\begin{equation}\label{eq:factory}
    P_{f,r}(x) = \left(\prod_{e \in E} x_e^{f_{e}} (1-x_{e})^{1-f_{e}}\right) \cdot \left(\sum_{T | \Flip_{f}(T) \in \Arb_{r}} \left(\prod_{e\in T}x_{e}^{1-f_{e}}(1-x_{e})^{f_e} \right)\right) .
\end{equation}

Observe that Algorithm \ref{alg:bernoulli_flow} samples each $f$ with probability proportionally to $P_{f,1}(p_1, \hdots, p_n)$. Since those are Bernstein polynomials, we can obtain an equivalent sampling procedure by plugging this polynomial in Lemma \ref{lemma:bernstein}. To show that each edge is sampled with the correct probabilities, it is enough to check the conditions in Lemma \ref{lemma:bernstein}.

\begin{theorem}\label{thm:cbf_cond}
Let $\cP$ be a flow-like polytope defined over a set of edges $E$ that are connected when viewed as undirected edges and $\cP \cap (0,1)^n \neq \emptyset$ and let $V$ be the set of its vertices. The polynomials $P_{f,r}(x)$ defined in equation \eqref{eq:factory} for each $f \in V$ and root $r \in [n]$ have the property that \begin{equation}\label{eq:cbf}
\sum_{f \in V} (f-x)  P_{f,r}(x) = 0, \forall x \in \cP
\end{equation}
\begin{equation}\label{eq:fnz}
\sum_{f \in V} P_{f,r}(x) > 0, \forall x \in \cP \cap (0,1)^n
\end{equation}
\end{theorem}

Theorem \ref{thm:main} is therefore a direct consequence of Theorem \ref{thm:cbf_cond}. The remainder of the paper will be dedicated to proving Theorem \ref{thm:cbf_cond}. Later in Section \ref{sec:discussion} we discuss the conditions in Theorem \ref{thm:cbf_cond} and how they can be relaxed.

%\section{Proof of Theorem \ref{thm:cbf_cond}}

\subsection{Root independence}

The first main observation is the remarkable fact that the values of the polynomials $P_{f,r}(x)$ for points $x \in \cP$ do not actually depend on the choice of the root $r$. More formally:

\begin{lemma}[Root Independence Lemma]\label{lemma:root-independence}
Let $\cP$ be a flow-like polytope $\cP$ and $P_{f,r}(x)$ be the polynomials defined as in equation \eqref{eq:factory} for each vertex $f$. Then given two distinct roots $r,s \in [n]$, then it holds that $P_{f,r}(x) = P_{f,s}(x)$ for all $x \in \cP$.
\end{lemma}

To prove the Root Independence Lemma, we will first define an auxiliary polynomial $\SArb_r$ over the affine span of the circulation polytope. For each root $r \in [n]$ we define  the function $\SArb_{r}: \CircBar \rightarrow \bbR$ as

$$\SArb_{r}(x) = \sum_{A \in \Arb_{r}} \prod_{e \in A} x_{e}.$$

Now, we show that as a consequence of the Matrix-Tree theorem, the value of $\SArb$  is independent of the choice of the root $r$:

\begin{lemma}\label{lem:sarb_root_independent}
Let $x \in \CircBar$, and $i, j \in [n]$. Then $\SArb_{i}(x) = \SArb_{j}(x)$.
\end{lemma}
\begin{proof}
Consider the $n$-by-$n$ matrix $M$ where $\cL_{ij} = -x_{ij}$ for $i \neq j$, and 
$$\cL_{ii} = \sum_{j \neq i} x_{ij} = \sum_{j \neq i} x_{ij}.$$
\noindent The second equality holds since $x \in \CircBar$. Note that $\cL$ is a  Laplacian matrix as define in Section \ref{sec:matrix_tree}. Also note that $\cL$ is a ZLS matrix since every row and column have zero sum. Let $\cL^{(i)}$ be the matrix obtained by removing the $i$-th row and $i$-th column from $\cL$. Note 
$$\SArb_i(x) = \det[\cL^{(i)}] = \det[\cL^{(j)}] = \SArb_j(x) $$
where the first and last equality follows from the Matrix-Tree Theorem (Theorem \ref{thm:tutte}) and the second equality follows from Theorem \ref{thm:zls} and the fact that $\cL$ is ZLS.
\end{proof}

Since $\SArb_{i}(x)$ is independent of $i$, we will omit $i$ and write $\SArb(x)$ to indicate the value of $\SArb_{i}(x)$ for any $i$. \\

The second ingredient for the proof of Lemma \ref{lemma:root-independence} will be a linear transformation mapping a flow-like polytope $\cP$ to $\CircBar$. Given a vertex $f$ of $\cP$ we will define the map $M_f:\bbR^E \rightarrow \bbR^E$ as follows:
$$M_{f}(x)_{e} = x_{e}(1 - f_{e}) + (1-x_{\overline{e}})f_{\overline{e}}.$$

The motivation for this particular map will become evident in the proof of Lemma \ref{lemma:root-independence}. Before we get there, we show the following lemma:

\begin{lemma} \label{lem:ptocirc}
Let $x$ be a point in $\cP$ and $f$ be a vertex of $\cP$. Then $M_{f}(x) \in \CircBar$.
\end{lemma}
\begin{proof}
We must show (for each $v \in [n]$) that

$$\sum_{i}M_{f}(x)_{vi} - \sum_{i}M_{f}(x)_{iv} = 0.$$

Expanding $M_{f}(x)$, this is equivalent to showing 

$$\sum_{i}\left(x_{vi}(1 - f_{vi}) + (1-x_{iv})f_{iv}\right) - \sum_{i}\left(x_{iv}(1 - f_{iv}) + (1-x_{vi})f_{vi}\right) = 0.$$

This is true since we can further simplify the LHS of the above expression as follows:

\begin{eqnarray*}
& & \sum_{i}\left(x_{vi}(1 - f_{vi}) + (1-x_{iv})f_{iv}\right) - \sum_{i}\left(x_{iv}(1 - f_{iv}) + (1-x_{vi})f_{vi}\right) \\
&=& \sum_{i}\left(x_{vi}(1 - f_{vi}) - (1-x_{vi})f_{vi}\right) - \sum_{i}\left(x_{iv}(1 - f_{iv}) - (1-x_{iv})f_{iv} \right) \\
&=& \sum_{i}\left(x_{vi} - f_{vi}\right) - \sum_{i}\left(x_{iv} - f_{iv} \right) \\
&=& \sum_{i}\left(x_{vi} - x_{iv}\right) - \sum_{i}\left(f_{vi} - f_{iv} \right) \\
&=& d_v - d_v = 0.
\end{eqnarray*}
\end{proof}

We are now ready to prove the main lemma:

% \begin{enumerate}
%     \item By the generalized matrix-tree theorem (see e.g. Theorem \ref{thm:tutte}), $\SArb_{i}(x)$ is equal the the $i$th minor of $M$ (the determinant of the $(n-1)$-by-$(n-1)$ matrix after removing the $i$th row and column from $M$).
%     \item Since $M$ is a ZLS matrix, then by Lemma \ref{thm:zls}, all principal minors of $M$ are equal (and equal to the product of the $(n-1)$ non-zero eigenvalues of $M$). 
% \end{enumerate}

% This implies that $\SArb_{i}(x) = \SArb_{j}(x)$, as desired.

% A potentially more useful way to think of this polynomial is as follows. We sum over directed trees $T$ that get sent to arborescences in $\Arb_{r}$ by the function $\Flip_f$. Each such tree contributes a single term (which ends up being a Bernstein monomial) to $P_{f}(x)$:

% \begin{itemize}
%     \item Edges $e$ which appear in $f$ (i.e., have $f_{e} = 1$) but not in $T$ each contribute a factor of $x_{e}$ to this term.
%     \item Edges $e$ which do not appear in $f$ or $T$ each contribute a factor of $(1-x_e)$ to this term.
%     \item Edges $e$ which appear in $T$ each contribute a factor of $x_{e}(1-x_e)$ to this term.
% \end{itemize}

% The goal for the rest of this section is to show that the right-hand side of \eqref{eq:factory} does not depend on the choice of root $r$.

\begin{proof}[Proof of Lemma \ref{lemma:root-independence}]

To do this, instead of expressing $P_{f,r}(x)$ as a sum over directed trees, we will find it more convenient to sum over all arborescences in $\Arb_{r}$. For each arborescence $A \in \Arb_{r}$ and vertex $f$ of $\cP$, consider the set $\Flip^{-1}_{f}(A)$ of directed trees which get sent to $A$ under $\Flip_{f}$. By doing this, we can rewrite \eqref{eq:factory} in the form:

$$P_{f}(x) = \left(\prod_{e \in E} x_e^{f_{e}} (1-x_{e})^{1-f_{e}}\right) \cdot \left(\sum_{A \in \Arb_{r}}\sum_{T \in \Flip^{-1}_{f}(A)} \left(\prod_{e\in T}x_{e}^{1-f_{e}}(1-x_{e})^{f_e} \right)\right).$$

Let us now consider the structure of $\Flip^{-1}_{f}(A)$. Since $\Flip_{f}(T)$ acts independently on each edge of $T$, $\Flip^{-1}_{f}(A)$ is the Cartesian product of $\Flip^{-1}_{f}(e)$ for all $e \in A$. In particular, we can further rewrite:

$$P_{f,r}(x) = \left(\prod_{e \in E} x_e^{f_{e}} (1-x_{e})^{1-f_{e}}\right) \cdot \left(\sum_{A \in \Arb_{r}}\prod_{e \in A}\left(\sum_{e' \in \Flip^{-1}_{f}(A)} x_{e'}^{1-f_{e'}}(1-x_{e'})^{f_{e'}} \right) \right).$$

We now categorize the possibilities for $\Flip_{f}^{-1}(e)$:

\begin{itemize}
    \item If $f_{e} = 0$ and $f_{\overline{e}} = 0$, then $\Flip_{f}^{-1}(e) = \{e\}$.
    \item If $f_{e} = 0$ and $f_{\overline{e}} = 1$, then $\Flip_{f}^{-1}(e) = \{e, \overline{e}\}$.
    \item If $f_{e} = 1$ and $f_{\overline{e}} = 0$, then $\Flip_{f}^{-1}(e) = \emptyset$.
    \item If $f_{e} = 1$ and $f_{\overline{e}} = 1$, then $\Flip_{f}^{-1}(e) = \{\overline{e}\}$.
\end{itemize}

A direct consequence of this categorization is that the following identity holds:

\begin{equation}\label{eq:invident}
\sum_{e' \in \Flip^{-1}_{f}(e)} x_{e'}^{1-f_{e'}}(1-x_{e'})^{f_{e'}} = x_{e}(1 - f_{e}) + (1-x_{\overline{e}})f_{\overline{e}}.
\end{equation}

The RHS is exactly the $e$-th component of $M_f(x)$, which allows us to re-write:

$$P_{f,r}(x) = \left(\prod_{e \in E} x_e^{f_{e}} (1-x_{e})^{1-f_{e}}\right) \cdot \left(\sum_{A \in \Arb_{r}}\prod_{e \in A} M_f(x)_e \right) =  \left(\prod_{e \in E} x_e^{f_{e}} (1-x_{e})^{1-f_{e}}\right) \cdot \SArb_r(M_f(x)).$$

By Lemma \ref{lem:ptocirc}, we know that $M_f(x) \in \CircBar$ and by Lemma \ref{lem:sarb_root_independent} we know that $\SArb$ is independent of the choice of the root for points in $\CircBar$. Taken together, those facts conclude the proof of the lemma.

\end{proof}

\subsection{Condition \eqref{eq:cbf}}

Now that we have established that the value of $P_{f,r}(x)$ for $x \in \cP$ doesn't depend on the choice of root $r$, we will drop $r$ from our notation and refer to these polynomials simply as $P_f(x)$. We will now establish condition \eqref{eq:cbf} in Theorem \ref{thm:cbf_cond}. This condition is a vector equality, so it is enough to show that each component is zero. 
We will fix a specific edge $\eta = (s,t) \in E$, and let $\eta = (s, t)$ and show that the $\eta$th component of the above sum in condition \eqref{eq:cbf} is $0$. Splitting this sum into two parts depending on the value of $f_{\eta}$, note that it suffices to show that

\begin{equation}\label{eq:fact_identity}
\sum_{f \in V; f_{\eta} = 1} (1 - x_{\eta}) P_{f}(x) = \sum_{f \in V;  f_{\eta} = 0} x_{\eta}P_{f}(x).
\end{equation}

\noindent Using Lemma \ref{lemma:root-independence}, we will write the left-hand side of \eqref{eq:fact_identity} as (rooting our arborescences at $t$)

$$\sum_{f_{\eta} = 1} (1 - x_{\eta}) \left(\prod_{e \in E} x_e^{f_{e}} (1-x_{e})^{1-f_{e}}\right) \cdot \left(\sum_{T | \Flip_{f}(T) \in \Arb_{t}} \left(\prod_{e\in T}x_{e}^{1-f_{e}}(1-x_{e})^{f_e} \right)\right).$$

\noindent
and the right-hand side of \eqref{eq:fact_identity} as (rooting our arborescences at $s$)

$$\sum_{f_{\eta} = 0} x_{\eta} \left(\prod_{e \in E} x_e^{f_{e}} (1-x_{e})^{1-f_{e}}\right) \cdot \left(\sum_{T | \Flip_{f}(T) \in \Arb_{s}} \left(\prod_{e\in T}x_{e}^{1-f_{e}}(1-x_{e})^{f_e} \right)\right).$$

Let us focus on the left-hand side momentarily. Note that if $f_{\eta} = 1$ and $\Flip_{f}(T) \in \Arb_{t}$, then it is impossible for $\eta$ to belong to $T$ (if $\eta \in T$, then $\Flip_{f}(T)$ would contain $\overline{\eta} = (t, s)$, and no arborescence rooted towards $t$ can contain the outgoing edge $(t, s)$). With this observation, we can rewrite the left-hand side as:

\begin{equation}\label{eq:lhs}\sum_{f_{\eta} = 1}\sum_{T | \Flip_{f}(T) \in \Arb_{t}} \left(\prod_{e \in E \setminus (T \cup \{\eta\})} x_e^{f_{e}} (1-x_{e})^{1-f_{e}}\right) \cdot \left(\prod_{e\in T \cup \{\eta\}}x_{e}(1-x_{e}) \right).\end{equation}

By similar logic (if $f_{\eta} = 0$ and $\Flip_{f}(T) \in \Arb_{s}$, then $\eta$ cannot belong to $T$), we can also rewrite the right-hand side as:

\begin{equation}\label{eq:rhs}\sum_{f_{\eta} = 0}\sum_{T | \Flip_{f}(T) \in \Arb_{s}} \left(\prod_{e \in E \setminus (T \cup \{\eta\})} x_e^{f_{e}} (1-x_{e})^{1-f_{e}}\right) \cdot \left(\prod_{e\in T \cup \{\eta\}}x_{e}(1-x_{e}) \right).\end{equation}

We will exchange the order of summation in the above two equations. For a directed tree $T \in \cT(E)$, let $\cF_{s}(T) \subseteq \cV(\cP)$ be the set of vertices $f$ of $\cP$ such that $\Flip_{f}(T) \in \Arb_{s}$ and $f_{\eta} = 0$. Similarly, let $\cF_{t}(T) \subseteq \cV(\cP)$ be the set of vertices $f$ of $\cP$ such that $\Flip_{f}(T) \in \Arb_{t}$ and $f_{\eta} = 1$. We can now write expression \eqref{eq:lhs} as

\begin{equation}\label{eq:lhs2}\sum_{T \in \cT(E)} \sum_{f \in \cF_{t}(T)} \left(\prod_{e \in E \setminus (T \cup \{\eta\})} x_e^{f_{e}} (1-x_{e})^{1-f_{e}}\right) \cdot \left(\prod_{e\in T \cup \{\eta\}}x_{e}(1-x_{e}) \right).\end{equation}

\noindent
and expression \eqref{eq:rhs} as

\begin{equation}\label{eq:rhs2}\sum_{T \in \cT(E)} \sum_{f \in \cF_{s}(T)} \left(\prod_{e \in E \setminus (T \cup \{\eta\})} x_e^{f_{e}} (1-x_{e})^{1-f_{e}}\right) \cdot \left(\prod_{e\in T \cup \{\eta\}}x_{e}(1-x_{e}) \right).\end{equation}

We will now show there exists a bijection between the monomials appearing in \eqref{eq:lhs2} and \eqref{eq:rhs2} and therefore
$$\sum_{f \in \cF_{s}(T)} \prod_{e \in E \setminus (T \cup \{\eta\})} x_e^{f_{e}} (1-x_{e})^{1-f_{e}} = \sum_{f \in \cF_{t}(T)} \prod_{e \in E \setminus (T \cup \{\eta\})} x_e^{f_{e}} (1-x_{e})^{1-f_{e}}.$$
concluding the proof. This bijection is established in the following lemma:

\begin{lemma}\label{lem:bijection}
There exists a bijection $\pi$ between $\cF_{s}(T)$ and $\cF_{t}(T)$ with the following property: if $\pi(f^{(s)}) = f^{(t)}$, then $f^{(s)}_{e} = f^{(t)}_{e}$ for all $e \in E \setminus (T \cup \{\eta\})$.
\end{lemma}

\begin{proof}
We will show that there is an element $g \in \CircBar$ such that $\pi(f) = f + g$ is such a bijection with the above property.

First, note that given any directed tree $T \in \cT$ and a vertex $v \in [n]$, there is a unique way to reorient the edges of $T$ so that it lies in $\Arb_{v}$ (in particular, all edges must be oriented towards $v$). Therefore there exists a unique arborescence $A_{s} \in \Arb_{s}$ such that if $\Flip_{f}(T) \in \Arb_{s}$ (for any $f$), then $\Flip_{f}(T) = A_{s}$. Similarly, there exists a unique arborescence $A_{t} \in \Arb_{t}$ such that if $\Flip_{f}(T) \in \Arb_{t}$, then $\Flip_{f}(T) = A_t$.

We now construct $g$ as follows. Set $g_{\eta} = 1$. Now, let $C^{+} = T \cap (A_{s} \setminus A_{t})$ be the set of edges in $T$ that belong to $A_{s}$ but not to $A_{t}$. Similarly, let $C^{-} = T \cap (A_{t} \setminus A_{s})$ be the set of edges in $T$ that belong to $A_{t}$ but not to $A_{s}$. For each edge $e \in C^{+}$, set $g_{e} = 1$, and for each edge $e \in C^{-}$, set $g_{e} = -1$. See Figure \ref{fig:proof} for an example.

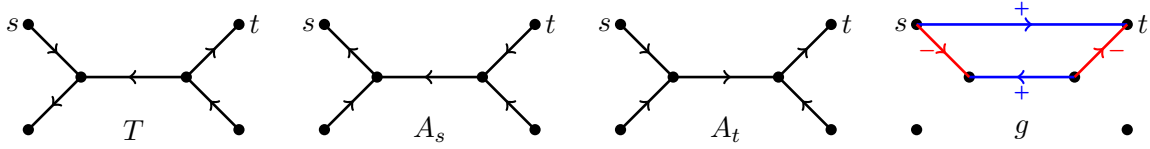
\begin{figure}[h]
\centering
\begin{tikzpicture}[scale=.7,decoration={
    markings,
    mark=at position 0.55 with {\arrow{>}}}
    ]

\node[circle,fill,inner sep=1.5pt] at (0,0) {};
\node[circle,fill,inner sep=1.5pt] at (0,2) {};
\node[circle,fill,inner sep=1.5pt] at (1,1) {};
\node[circle,fill,inner sep=1.5pt] at (3,1) {};
\node[circle,fill,inner sep=1.5pt] at (4,0) {};
\node[circle,fill,inner sep=1.5pt] at (4,2) {};
\draw[postaction={decorate},line width=1pt] (0,2) -- (1,1);
\draw[postaction={decorate},line width=1pt] (1,1) -- (0,0);
\draw[postaction={decorate},line width=1pt] (3,1) -- (1,1);
\draw[postaction={decorate},line width=1pt] (4,0) -- (3,1);
\draw[postaction={decorate},line width=1pt] (3,1) -- (4,2);
\node at (2,0) {$T$};
\node at (-.3,2) {$s$};
\node at (4.3,2) {$t$};

\begin{scope}[xshift=160]
\node[circle,fill,inner sep=1.5pt] at (0,0) {};
\node[circle,fill,inner sep=1.5pt] at (0,2) {};
\node[circle,fill,inner sep=1.5pt] at (1,1) {};
\node[circle,fill,inner sep=1.5pt] at (3,1) {};
\node[circle,fill,inner sep=1.5pt] at (4,0) {};
\node[circle,fill,inner sep=1.5pt] at (4,2) {};
\draw[postaction={decorate},line width=1pt] (1,1) -- (0,2);
\draw[postaction={decorate},line width=1pt] (0,0) -- (1,1);
\draw[postaction={decorate},line width=1pt] (3,1) -- (1,1);
\draw[postaction={decorate},line width=1pt] (4,0) -- (3,1);
\draw[postaction={decorate},line width=1pt] (4,2) -- (3,1);
\node at (2,0) {$A_s$};
\node at (-.3,2) {$s$};
\node at (4.3,2) {$t$};
\end{scope}

\begin{scope}[xshift=320]
\node[circle,fill,inner sep=1.5pt] at (0,0) {};
\node[circle,fill,inner sep=1.5pt] at (0,2) {};
\node[circle,fill,inner sep=1.5pt] at (1,1) {};
\node[circle,fill,inner sep=1.5pt] at (3,1) {};
\node[circle,fill,inner sep=1.5pt] at (4,0) {};
\node[circle,fill,inner sep=1.5pt] at (4,2) {};
\draw[postaction={decorate},line width=1pt] (0,2) -- (1,1);
\draw[postaction={decorate},line width=1pt] (0,0) -- (1,1);
\draw[postaction={decorate},line width=1pt] (1,1) -- (3,1);
\draw[postaction={decorate},line width=1pt] (4,0) -- (3,1);
\draw[postaction={decorate},line width=1pt] (3,1) -- (4,2);
\node at (2,0) {$A_t$};
\node at (-.3,2) {$s$};
\node at (4.3,2) {$t$};
\end{scope}

\begin{scope}[xshift=480]
\node[circle,fill,inner sep=1.5pt] at (0,0) {};
\node[circle,fill,inner sep=1.5pt] at (0,2) {};
\node[circle,fill,inner sep=1.5pt] at (1,1) {};
\node[circle,fill,inner sep=1.5pt] at (3,1) {};
\node[circle,fill,inner sep=1.5pt] at (4,0) {};
\node[circle,fill,inner sep=1.5pt] at (4,2) {};
\draw[blue, postaction={decorate},line width=1pt] (0,2) -- (4,2);
\node at (2,2.3) {{\color{blue}$+$}};
\draw[blue, postaction={decorate},line width=1pt] (3,1) -- (1,1);
\node at (2,0.7) {{\color{blue} $+$}};

\draw[red, postaction={decorate},line width=1pt] (0,2) -- (1,1);
\draw[red, postaction={decorate},line width=1pt] (3,1) -- (4,2);
\node at (0.2,1.5) {{\color{red} $-$}};
\node at (3.8,1.5) {{\color{red} $-$}};
\node at (2,0) {$g$};
\node at (-.3,2) {$s$};
\node at (4.3,2) {$t$};
\end{scope}

\end{tikzpicture}
\caption{A directed tree $T$ and its corresponding arborescences $A_s$ and $A_t$ that can be obtained by flipping edges so that all edges are oriented towards one of them. From $T$, $A_s$ and $A_t$ we can construct $g$ which is a collection of edges, some with $+1$ weight (in blue) and some with $-1$ weight (in red). }
\label{fig:proof}
\end{figure}

\begin{figure}[h]
\centering
\begin{tikzpicture}[scale=.7,decoration={
    markings,
    mark=at position 0.55 with {\arrow{>}}}
    ]

\node[circle,fill,inner sep=1.5pt] at (0,0) {};
\node[circle,fill,inner sep=1.5pt] at (0,2) {};
\node[circle,fill,inner sep=1.5pt] at (1,1) {};
\node[circle,fill,inner sep=1.5pt] at (3,1) {};
\node[circle,fill,inner sep=1.5pt] at (4,0) {};
\node[circle,fill,inner sep=1.5pt] at (4,2) {};
\draw[postaction={decorate},line width=1pt] (0,2) -- (1,1);
\draw[postaction={decorate},line width=1pt] (1,1) -- (0,0);
\draw[postaction={decorate},line width=1pt] (0,0) -- (3,1);
\draw[postaction={decorate},line width=1pt] (3,1) -- (4,2);
\draw[postaction={decorate},line width=1pt] (4,2) -- (0,2);
\node at (2,0) {$f$};
\node at (-.3,2) {$s$};
\node at (4.3,2) {$t$};

\begin{scope}[xshift=200]
\node[circle,fill,inner sep=1.5pt] at (0,0) {};
\node[circle,fill,inner sep=1.5pt] at (0,2) {};
\node[circle,fill,inner sep=1.5pt] at (1,1) {};
\node[circle,fill,inner sep=1.5pt] at (3,1) {};
\node[circle,fill,inner sep=1.5pt] at (4,0) {};
\node[circle,fill,inner sep=1.5pt] at (4,2) {};
% \draw[postaction={decorate},line width=1pt] (0,2) -- (1,1);
\draw[postaction={decorate},line width=1pt] (1,1) -- (0,0);
\draw[postaction={decorate},line width=1pt] (0,0) -- (3,1);
% \draw[postaction={decorate},line width=1pt] (3,1) -- (4,2);
\draw[postaction={decorate},line width=1pt] (4,2) -- (0,2);
\draw[postaction={decorate},line width=1pt] (3,1) -- (1,1);
\draw[postaction={decorate},line width=1pt] (0,2) to [out=30,in=150] (4,2);
\node at (2,0) {$f+g$};
\node at (-.3,2) {$s$};
\node at (4.3,2) {$t$};
\end{scope}

\end{tikzpicture}
\caption{For $\cP = \Circ$, we have a flow $f$ such that $\Flip_f(T) = A_s$ for $T$ and $A_s$ as in Figure \ref{fig:proof}. When we sum $f+g$ we obtain a flow such that $\Flip_{f+g}(T)=A_t$.}
\label{fig:proof2}
\end{figure}
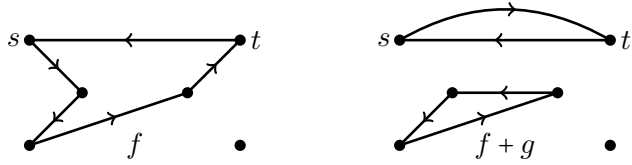

\paragraph{Showing $f+g \in \{0, 1\}^{E}$.} To prove that such a $g$ works, fix an $f \in \cF_{s}(T)$ and let $f' = \pi(f) = f + g$. We will first show that $f' \in \{0, 1\}^E$. Since $f \in \{0, 1\}^E$, it suffices to show that for all edges $e$ where $g_{e} = 1$, that $f_{e} = 0$, and for all edges $e$ where $g_{e} = -1$, that $f_{e} = -1$.

Edges $e \in C^{+}$ belong to both $T$ and $A_{s}$. If $f_{e} = 1$, then $e$ would be flipped in (and not belong to) $\Flip_{f}(T)$, so there is no way this could equal $A_{s}$ (contradicting the fact that $f \in \cF_{s}(T)$). It follows that $f_{e} = 0$ for $e \in C^{+}$ (where $g_{e} = 1$). Similarly, edges $e \in C^{-}$ belong to $T$ but do not belong to $A_{s}$. If $f_{e} = 0$, then $e$ would not be flipped in (and thus belong to) $\Flip_{f}(T)$, so again there is no way this could equal $A_{s}$. It follows that $f_{e} = 1$ for $e \in C^{-}$ (where $g_{e} = -1$). Finally, by the definition of $\cF_{s}(T)$, $f_{\eta} = 0$ (so $g_{\eta} = 1$ is fine). 

\paragraph{Showing $f+g \in \cP$.} We have now shown $f' \in \{0, 1\}^{E}$. To show that $f'$ is in fact a vertex of $\cP$, it suffices to show that $g \in \CircBar$ (and therefore that $f+g$ lies on the hyperplane \eqref{eq:affine} defining $\cP$). To show that $g \in \CircBar$, we will write $g$ as a linear combination of indicators of directed cycles. Let $C = \eta \cup C^{+} \cup \overline{C^{-}}$ (where $\overline{C^{-}} = \{\overline{e}; e\in C^{-}\}$) -- we claim $C$ is a directed cycle. In particular, we claim that $C$ is the directed cycle formed by $\eta$ and the directed path from $t$ to $s$ along the edges in $T$ (when viewed as an undirected tree). 

To see this, note that $A_{s} \setminus A_{t}$ consists exactly of the directed path from $t$ to $s$, and $A_{t} \setminus A_{s}$ consists exactly of the directed path from $s$ to $t$. $C^{+}$ therefore contains the edges in $T$ that are oriented towards $s$ on the directed path from $t$ to $s$, and $C^{-}$ contains the edges in $T$ that oriented towards $t$ on the directed path from $s$ to $t$ (and when reversed, form the missing edges on the directed path from $t$ to $s$).

Now, let $c \in \CircBar$ be the indicator of the directed cycle $C$, and given an edge $e$, let $\nu(e)$ be the indicator of the two edge directed cycle $\{e, \overline{e}\}$. We then have that

$$g = c - \sum_{e \in C^{-}} \nu(e),$$

\noindent
and therefore $g \in \CircBar$.

\paragraph{Showing $f+g \in \cF_{t}(T)$.} We next argue that $\Flip_{f'}(T) = A_{t}$. Note that for all edges $e \in A_{s} \cap A_{t}$, $f'_{e} = f_{e}$, so if $\Flip_{f}(T)$ contains $e$, then $\Flip_{f'}(T)$ also contains $e$. On the other hand, for each edge $e \in A_{s} \triangle A_{t}$, either $e \in T$ or $\overline{e} \in T$. If $e \in T$, then (since $g_{e} \neq 0$) $f_{e} \neq f'_{e}$, so $\Flip_{f}(T)$ and $\Flip_{f'}(T)$ will contain different orientations of $e$. Since $\Flip_{f}(T)$ contains the orientation of $e$ belonging to $A_{s}$, $\Flip_{f'}(T)$ contains the orientation of $e$ belonging to $A_{t}$. The same argument works if $\overline{e} \in T$, replacing $e$ with $\overline{e}$ in the preceding sentence.

Since $f' \in \cV(\cP)$, $\Flip_{f'}(T) = A_{t}$, and $f'_{\eta} = 1$, it follows that $f' \in \cF_{t}(T)$. We can see that $\pi$ is a bijection with inverse given by $\pi^{-1}(f') = f' - g$ (a symmetric argument to that given above shows that this maps any $f' \in \cF_{t}(T)$ to an element of $\cF_{s}(T)$). Finally, note that since the support of $g$ is contained in $T \cup \{\eta\}$, $f$ and $f'$ agree on all edges not in $T \cup \{\eta\}$, so the additional restriction on the bijection is true. 
\end{proof}

% By Lemma \ref{lem:bijection}, the expressions in \eqref{eq:lhs2} and \eqref{eq:rhs2} are equal; in particular, Lemma \ref{lem:bijection} directly implies that for any $T \in \cT(E)$,

% $$\sum_{f \in \cF_{s}(T)} \prod_{e \in E \setminus (T \cup \{\eta\})} x_e^{f_{e}} (1-x_{e})^{1-f_{e}} = \sum_{f \in \cF_{t}(T)} \prod_{e \in E \setminus (T \cup \{\eta\})} x_e^{f_{e}} (1-x_{e})^{1-f_{e}}.$$

% \noindent
% It follows that the polynomials $P_{f}(x)$ generate a valid combinatorial Bernoulli factory for $\cP$. 

\subsection{Condition \eqref{eq:fnz}}

One issue that could occur with the above method is that we may obtain a collection of polynomials $P_f(x)$ which satisfy \eqref{eq:cbf} but all equal $0$. This is prevented by condition \eqref{eq:fnz} in Theorem \ref{thm:cbf_cond} which requires the polynomials to be a non-trivial solution.

% Since this Bernoulli factory outputs a vertex $f \in \cV(\cP)$ with probability $P_{f}(x)/\sum_{f' \in \cV(\cP)} P_{f'}(x)$ (and does so by running a Bernoulli race among simulated coins with probabilities given by these polynomials), this is a problem.

\begin{lemma}\label{thm:positive_factory}
Let $\cP$ be a flow-based polytope where the set of variable edges $E$ is connected and $\cP \cap (0, 1)^n \neq \emptyset$. Then there exists a flow $f \in \cV(\cP)$ and a directed tree $T \in \cT(E)$ such that $\Flip_{f}(T)$ is an arborescence (rooted at some node $r \in [n]$). 
\end{lemma}

Note that by the construction of $P_{f}(x)$ in \eqref{eq:factory}, Lemma \ref{thm:positive_factory} implies $P_{f}(x) > 0$ for any $x \in \cP \cap (0, 1)^n$, and therefore that a Bernoulli race run over the $P_{f}(x)$ will terminate in finite time. An interesting consequence of Equation \eqref{eq:factory} and the fact that it is invariant to the choice of root is that it does not matter which vertex the arborescence $\Flip_{f}(T)$ is rooted at (that is, for every root $r$, there exists an arborescence rooted at $r$ of the form $\Flip_{f}(T)$).

To prove Lemma \ref{thm:positive_factory}, we will need to make use of the following two lemmas. The first characterizes directed graphs where there exists an arborescence rooted at every node.

\begin{lemma}\label{lem:pos1}
A directed graph has an arborescence rooted at every node iff it is strongly connected.
\end{lemma}
\begin{proof}
If the directed graph has an arborescence $A_r$ rooted at each node $r \in [n]$, you can always get from any node $i$ to any node $j$ by following the path from $i$ to $j$ in $A_r$. 

If the directed graph is strongly connected, to get an arborescence rooted at $r$, choose a spanning arborescence of the union of all paths from nodes $i \in [n]$ to $r$.
\end{proof}

The second lemma gives an alternate requirement for a directed graph to be strongly connected in terms of circulations of the graph.

\begin{lemma}\label{lem:pos2}
Assume a directed graph $G$ has the property that every edge belongs to some circulation only using edges of $G$ and that the edges of $G$ form a connected graph when undirected. Then $G$ is strongly connected.
\end{lemma}
\begin{proof}
We will show there is a directed path from $s$ to $t$ for any $s, t \in [n]$. Since $G$ is connected, consider the undirected path $P$ of edges from $s$ to $t$. For each edge $e$ in the path $P$: 

\begin{enumerate}
    \item if $e$ is directed towards $t$, then follow the edge $e$.
    \item if $e$ is directed towards $s$, then look at the circulation $C$ that $e$ belongs to. $C$ can be decomposed into a set of edge-disjoint cycles, one of which contains $e$. We start at the target of $e$ and follow this cycle until we get to the source of $e$ (this has the net effect of following the reversal of $e$).
\end{enumerate}

This sequence of edges allows us to reach $t$ from $s$, as desired.
\end{proof}

Note that Lemmas \ref{lem:pos1} and \ref{lem:pos2} directly imply Lemma \ref{thm:positive_factory} in the case where demands are all zero (i.e., $\cP$ is a polytope of circulations). In particular, the condition $\cP \cap (0, 1)^n \neq \emptyset$ implies that each edge belongs to at least one circulation (since for a point $x \in \cP \cap (0, 1)^n$, each edge has a positive probability of being selected by a valid factory), so by Lemma \ref{lem:pos2} the set of variable edges is strongly connected, and therefore by Lemma \ref{lem:pos1} the set of variable edges $E$ contains an arborescence rooted at every node. But since the empty flow (where $x_e = 0$ for all $e \in E$) is a circulation in $\cP$, this implies Lemma \ref{thm:positive_factory}.

We will show a generalization of this approach can be used to prove Lemma \ref{thm:positive_factory}.

\begin{proof}[Proof of Lemma~\ref{thm:positive_factory}]
Fix any specific flow $f \in \cV(\cP)$, and let $E' = \Flip_{f}(E)$. We will show that if an edge $e \in E$ appears in some flow $f' \in \cV(\cP)$, then the edge $e' = \Flip_{f}(e)$ appears in some circulation using the edges in $E'$. From this it will follow from the two above lemmas that $E'$ contains an arborescence rooted at every node (and hence there exists a tree $T \in \cT(E)$ such that $\Flip_{f}(T)$ is an arborescence.

Let $f'$ be a flow in $\cV(\cP)$. We will construct a circulation $c$ using edges in $E'$ as follows. Start by constructing the vector $c = f' - f \in [-1, 1]^{E \cup E'}$; note that $c \in \CircBar$ (e.g. by Lemma \ref{lem:ptocirc}). Next, we will transform $c$ to be supported only on the edge set $E'$ by doing the following: for any edge $e \in E$ where $c_e = -1$, add $1$ to both $c_{e}$ and $c_{\overline{e}}$ (note that if $c_{e} = -1$, then $e \in f$, so $\overline{e}$ belongs to $E'$). This maintains the property that $c \in \CircBar$, but now $c \in [0, 2]^{E'}$. Finally, this means that $c' = c/2$ is a proper element of $\Circ$ and is supported only on the elements of $E'$. As a consequence of this, for any $e \in E'$ where $c'_{e} > 0$, there exists a circulation in $E'$ containing the directed edge $e$.

Now, consider any edge $e \in E'$. If $e \not\in f$ (and thus $e \in E$), then if we can find a flow $f'$ such that $e \in f'$, it will be the case that $c'_e > 0$ and thus that there exists a circulation in $E'$ using the edge $e$. On the other hand, if $e \in f$, then since $e$ is also in $E'$, we must also have $\overline{e} \in E$. Then, if we can find a flow $f'$ such that $\overline{e} \not\in f'$, it will be the case that $c'_e > 0$, and thus that there exists a circulation in $E'$ using the edge $e$. But note that since $\cP \cap (0, 1)^n$ is non-empty, for any edge $e \in E$ we can find both a flow using $e$ and a flow not using $e$ (since a factory for an $x \in \cP \cap (0, 1)^n$ must have a positive probability of outputting both types of flows). It follows that for each edge in $E'$, there is a circulation using edges in $E'$ containing that edge. Since $E$ and $E'$ share the same undirected edge set, $E'$ is also connected as a set of undirected edges, so by Lemma \ref{lem:pos2} $E'$ is strongly connected, and by Lemma \ref{lem:pos1}, $E'$ has an arborescence rooted at every node. For any such arborescence $A$, $\Flip_{f}^{-1}(A)$ is a directed tree using edges in $E$ satisfying the conditions of this theorem.
\end{proof}

\subsection{Relaxing the Conditions in Theorem \ref{thm:main}} \label{sec:discussion}

Theorem \ref{thm:main} has a few restrictions that avoid some subtle issues that could otherwise occur. Here we discuss them in more detail:

\paragraph{Coins have $p_e \in (0,1)$.} We assume that no coin is deterministically zero or one.  Note that while our polynomials are guaranteed to be strictly positive for $p \in (0,1)^n$ they will vanish at certain points in the boundary. This issue is unavoidable even for the simple problem of sampling a $k$-subset, which can be viewed as the problem of sampling a $k$-flow in a graph with only two nodes (source and sink) and $n$ parallel edges. For $k$-subset, \cite{niazadeh2021combinatorial} shows that there exists no Bernoulli factories based on Bernstein polynomials (and more generally no exponentially-converging factories) that extend to the boundary of the hypercube. We also note that the assumption that the coins have biases strictly in $(0,1)$ is common in the Bernoulli factory literature and is made in the original paper of \cite{keane1994bernoulli} as well all most recent papers such as \cite{mossel2005new} and \cite{nacu2005fast}.

\paragraph{Non-empty interior.} For simplicity we assume that $\cP \cap (0, 1)^n \neq \emptyset$, but this condition can be easily relaxed. If $\cP$ is contained in a face of $[0, 1]^n$ and some of the variables $x_{e}$ are identically $0$ or $1$ in all of $\cP$. In this case, instead of constructing a Bernoulli factory for $\cP$, we can construct a Bernoulli factory for a lower-dimensional flow-based polytope $\cP'$ formed by eliminating these variables and adjusting demands.

\paragraph{Disconnected edge set.} If the edge set $E$ is disconnected there are no directed trees that only use the edges in $E$, so the polynomials in equation \eqref{eq:factory} must always evaluate to zero. But in this case, we can decompose $E$ into its connected components, construct a factory for the flow-based polytope induced by each of these connected components, and sample from all these subfactories to generate one sample for the original flow-based polytope. 

\paragraph{Efficiently sampling a directed tree.} Actually executing Algorithm \ref{alg:bernoulli_flow} requires one to (for some fixed flow $f$) uniformly sample a directed tree $T$ with edges in $E$ such that $\Flip_{f}(T) \in \Arb_1$. We describe here why this is possible to do in polynomial time. In brief, this follows from the fact that it is possible to count such trees via the Matrix-Tree Theorem. In more detail, note that sampling such a tree $T$ is equivalent to sampling an arborescence rooted at $1$ in the multigraph $E' = \Flip_{f}^{-1}(E)$. But it is possible to count the number of arborescences rooted at a specific node by the Matrix-Tree Theorem (Theorem \ref{thm:tutte}); moreover, this problem is self-reducible (i.e., we can count the number of arborescences that use a specific subset of edges by contracting these edges), and therefore it is also possible to uniformly sample such an arborescence. Alternatively, one can directly use a procedure such as that in \cite{anari2021sampling}.

\bibliographystyle{plainnat}
\bibliography{bernoulli}

\appendix

\section{Proof of Lemma \ref{lem:bernstein}}\label{appendix:bernstein}

Lemma \ref{lem:bernstein} is proved in Section 2 of \cite{niazadeh2021combinatorial}, but we provide a full proof here for completeness. We first two simple auxiliary lemmas.

\begin{lemma}\label{lem:sampling_monomials}
Let $M(x) = \prod_i x_i^{a_i} (1-x_i)^{b_i}$ be a Bernstein monomial with non-negative integer coefficients $a_i, b_i$. Then given access to coins with biases $x_1, \hdots, x_n$ we can sample from a coin with bias $M(x)$.
\end{lemma}

\begin{proof}
Flip the $x_i$-coin $a_i+b_i$ times. If for every coin $i$ the first $a_i$ coin flips are $1$ and the remaining coin flips are $0$, then we output $1$. Otherwise we output $0$. The probability we output $1$ is $\prod_i x_i^{a_i} (1-x_i)^{b_i} = M(x)$.
\end{proof}

\begin{lemma}\label{lem:sampling_polynomials}
Let $P(x) = \sum_i c_i M_i(x)$ be a Bernstein polynomial where $c_i \geq 0$ are real coefficients such that $\sum_i c_i \leq 1$ and $M_i(x)$ are Bernstein monomials as in Lemma \ref{lem:sampling_monomials}. Then given access to coins with biases $x_1, \hdots, x_n$ we can sample from a coin with bias $P(x)$.
\end{lemma}

\begin{proof}
If the sum of coefficients is less than $1$, pretend there exists an extra index $\emptyset$ with weight $c_{\emptyset} = 1 - \sum_i c_i$. Now, use external randomness to sample an index with probabilty $c_i$. If we sample $\emptyset$, return 0. Otherwise, we use Lemma \ref{lem:sampling_monomials} to sample from a coin with bias $M_i(x)$. The probability we sample $1$ is $\sum_i c_i M_i(x) = P(x)$.
\end{proof}

\begin{proof}[Proof of Lemma \ref{lem:bernstein}]
We apply the technique known as the Bernoulli race of \cite{dughmi2017bernoulli} in conjunction with Lemma \ref{lem:sampling_polynomials}. Before we describe it, observe that we can re-scale all polynomials such that the sum of coefficients in each of them is at most $1$ since the ratio $P_v(x) / \sum_{v'} P_{v'}(x)$ doesn't change by re-scaling.

We sample an index $v \in V$ uniformly at random using external randomness. The we use the procedure in Lemma \ref{lem:sampling_polynomials}. Then flip a coin of bias $P_v(x)$. If the coin is $1$, we output $v$. Otherwise we re-try with another $v$ chosen uniformly at random. It is simple to see that we choose each $v$ with probability proportional to $P_v(x) / \abs{V}$. Since this procedure outputs almost surely the probability we choose each index is $P_v(x) / \sum_{v'} P_{v'}(x)$ as desired.
\end{proof}

\end{document}